\documentclass[prl,superscriptaddress, twocolumn]{revtex4}

\usepackage{amsfonts,amssymb,amsmath}
\usepackage[]{graphics,graphicx,epsfig}
\usepackage{amsthm}
\usepackage{enumitem}

\def\identity{\leavevmode\hbox{\small1\kern-3.8pt\normalsize1}}

\newtheorem*{mydef}{Definition}

\renewcommand{\epsilon}{\varepsilon}

\newtheorem{definition}{Definition}

\newtheorem{lemma}[definition]{Lemma}

\newtheorem{thm}[definition]{Theorem}

\newtheorem*{rep@theorem}{\rep@title}
\newcommand{\newreptheorem}[2]{%
\newenvironment{rep#1}[1]{%
 \def\rep@title{#2 \ref{##1} (restatement)}%
 \begin{rep@theorem}}%
 {\end{rep@theorem}}}
\makeatother

\newtheorem{example}{Example}
\newtheorem{fact}{Fact}
\newreptheorem{thm}{Theorem}
\newreptheorem{lem}{Lemma}

\def\ba#1\ea{\begin{align}#1\end{align}}
\def\ban#1\ean{\begin{align*}#1\end{align*}}

\newcommand{\be}{\begin{equation}}
\newcommand{\ee}{\end{equation}}

\def\benum{\begin{enumerate}}
\def\eenum{\end{enumerate}}

\def\squareforqed{\hbox{\rlap{$\sqcap$}$\sqcup$}}
\def\qed{\ifmmode\squareforqed\else{\unskip\nobreak\hfil
\penalty50\hskip1em\null\nobreak\hfil\squareforqed
\parfillskip=0pt\finalhyphendemerits=0\endgraf}\fi}
\def\endenv{\ifmmode\;\else{\unskip\nobreak\hfil
\penalty50\hskip1em\null\nobreak\hfil\;
\parfillskip=0pt\finalhyphendemerits=0\endgraf}\fi}

\newcommand{\<}{\langle}
\renewcommand{\>}{\rangle}

\def\be{\begin{equation}}
\def\ee{\end{equation}}
\def\ben{\begin{eqnarray}}
\def\een{\end{eqnarray}}

\def\bei{\begin{itemize}}
\def\eei{\end{itemize}}

\mathchardef\ordinarycolon\mathcode`\:
\mathcode`\:=\string"8000
\def\vcentcolon{\mathrel{\mathop\ordinarycolon}}
\begingroup \catcode`\:=\active
  \lowercase{\endgroup
  \let :\vcentcolon
  }

\newcommand{\nc}{\newcommand}
 \nc{\proj}[1]{|#1\rangle\!\langle #1 |} 
\nc{\avg}[1]{\langle#1\rangle}

\nc{\conv}{\operatorname{conv}}
\nc{\smfrac}[2]{\mbox{$\frac{#1}{#2}$}} \nc{\Tr}{\operatorname{Tr}}
\nc{\ox}{\otimes} \nc{\dg}{\dagger} \nc{\dn}{\downarrow}
\nc{\lmax}{\lambda_{\text{max}}}
\nc{\lmin}{\lambda_{\text{min}}}

\nc{\csupp}{{\operatorname{csupp}}}
\nc{\qsupp}{{\operatorname{qsupp}}} \nc{\var}{\operatorname{var}}
\nc{\rar}{\rightarrow} \nc{\lrar}{\longrightarrow}
\nc{\poly}{\operatorname{poly}}
\nc{\polylog}{\operatorname{polylog}} \nc{\Lip}{\operatorname{Lip}}
\nc{\Om}{\Omega}
\nc{\wt}[1]{\widetilde{#1}}

\def\>{\rangle}
\def\<{\langle}

\nc{\glneq}{{\raisebox{0.6ex}{$>$}  \hspace*{-1.8ex} \raisebox{-0.6ex}{$<$}}}
\nc{\gleq}{{\raisebox{0.6ex}{$\geq$}\hspace*{-1.8ex} \raisebox{-0.6ex}{$\leq$}}}

\nc{\vholder}[1]{\rule{0pt}{#1}}
\nc{\wh}[1]{\widehat{#1}}
\nc{\h}[1]{\widehat{#1}}

\nc{\ob}[1]{#1}

\def\beq{\begin {equation}}
\def\eeq{\end {equation}}

\def\be{\begin{equation}}
\def\ee{\end{equation}}

\nc{\eq}[1]{(\ref{eq:#1})} 
\nc{\eqs}[2]{\eq{#1} and \eq{#2}}

\nc{\eqn}[1]{Eq.~(\ref{eqn:#1})}
\nc{\eqns}[2]{Eqs.~(\ref{eqn:#1}) and (\ref{eqn:#2})}

\nc{\region}{\cS\cW}

\newenvironment{protocol*}[1]
  {
    \begin{center}
      \hrulefill\\
      \textbf{#1}
  }
  {
    \vspace{-1\baselineskip}
    \hrulefill
    \end{center}
  }
\newcommand{\defeq}{\vcentcolon=}

\begin{document}

\title{No quantum realization of extremal no-signaling boxes}

\author{Ravishankar \surname{Ramanathan}}
\email{ravishankar.r.10@gmail.com}
\affiliation{National Quantum Information Center of Gda\'{n}sk,  81-824 Sopot, Poland}
\affiliation{Institute of Theoretical Physics and Astrophysics, University of Gda\'{n}sk, 80-952 Gda\'{n}sk, Poland}
\author{Jan Tuziemski}
\affiliation{National Quantum Information Center of Gda\'{n}sk, 81-824 Sopot, Poland}
\affiliation{Faculty of Applied Physics and Mathematics, Technical University of Gda\'{n}sk, 80-233 Gdansk, Poland}
\author{Micha{\l} Horodecki}
\affiliation{National Quantum Information Center of Gda\'{n}sk, 81-824 Sopot, Poland}
\affiliation{Institute of Theoretical Physics and Astrophysics, University of Gda\'{n}sk, 80-952 Gda\'{n}sk, Poland}
\author{Pawe{\l} \surname{Horodecki}}
\affiliation{National Quantum Information Center of Gda\'{n}sk, 81-824 Sopot, Poland}
\affiliation{Faculty of Applied Physics and Mathematics, Technical University of Gda\'{n}sk, 80-233 Gdansk, Poland}


\begin{abstract}
The study of quantum correlations is important for fundamental reasons as well as for quantum communication and information processing tasks. On the one hand, it is of tremendous interest to \textit{derive} the correlations produced by measurements on separated composite quantum systems from within the set of all correlations obeying the no-signaling principle of relativity, by means of information-theoretic principles \cite{Principles1, Principles2, Principles3, Principles4, Principles5}. On the other hand, cryptographic protocols based on quantum non-local correlations have been proposed for the generation of secure keys \cite{BHK} and the amplification and expansion of randomness \cite{CR, Pir10} against general no-signaling adversaries. In both these research programs, a fundamental question arises : can any measurements on quantum states realize the correlations present in pure extremal no-signaling boxes?   
Here, we answer this question in full generality showing that \textit{no non-trivial} (not local realistic) extremal boxes of general no-signaling theories can be realized in quantum theory. We then explore some important consequences of this fact.
\end{abstract}

\maketitle

It is well-known that quantum mechanics allows for non-local correlations between spatially separated systems, i.e., correlations that cannot be explained by any local hidden variable theory as shown by the violation of Bell inequalities. The correlations described by quantum theory form a convex set which is sandwiched between the sets of classical correlations and general no-signaling correlations. By how much and why the quantum set is inside the no-signaling set has been the subject of intense research \cite{Tsirelson, Principles1, Principles2, Principles3, Principles4, Principles5}.

The classical and no-signaling sets are known to form convex polytopes, the quantum set while being convex is not a polytope in general. Each point in the set describes a \textit{box} of correlations, that is a set of conditional probability distributions for outputs given inputs. The extremal boxes (vertices) of the no-signaling polytope are the \textit{pure} states of the no-signaling theory and their purity implies that they are completely uncorrelated with the environment. Consequently, access to an extremal non-local box would be of great advantage in cryptographic tasks, providing intrinsic certified randomness and key secure against any eavesdropper even when the latter is limited only by the no-signaling principle. Extremal boxes are also of crucial importance in the program of \textit{deriving} the set of quantum correlations from purely information-theoretic principles, being essential test-beds for checking the validity of such principles \cite{Principles1, Principles2, Principles3, Principles4, Principles5}. Furthermore, the extremal boxes are the ones that maximally violate Bell inequalities in no-signaling theories; for some Bell inequalities there exists a single \textit{unique} box that gives maximal violation. For instance, for the celebrated CHSH inequality \cite{CHSH}, there exists a unique extremal box known as the Popescu-Rohrlich (PR) box \cite{PR} which maximally violates the inequality. From these considerations, one arrives at the fundamental question: does quantum mechanics allow for the realization of extremal non-local no-signaling boxes? In other words, is there a Bell scenario (for some number of parties, inputs and outputs) for which the quantum set reaches a non-local vertex of the corresponding no-signaling polytope? As far as we are aware, this question was first posed formally in \cite{Fritz2} following the results of \cite{YCATS} where it was shown that for three parties with two inputs and two outputs, no quantum realization of extremal no-signaling boxes is possible. 

Here, we show that the answer to the question is negative in full generality: \textit{no quantum realization of an extremal non-local no-signaling box exists in any Bell scenario}. To show this, we make use of recently discovered connections between contextuality, Bell scenarios and graph theory \cite{Winter, CSW, Principles4, Principles5, Fritz}. We then consider a few interesting consequences of this fact. Firstly, as a simple application, we consider the important class of Bell inequalities known as multi-player \textit{unique games} and give a simple condition for identifying when the inequality is maximally violated by a single vertex. Our main result coupled with the fact that all unique games can be won with certainty in general no-signaling theories, then shows simply that quantum theory does not allow for a winning strategy for these unique games. 

Secondly, the non-local vertices being excluded from quantum theory implies that any box that can be distilled to such a vertex by operations (wirings) performed locally by the parties is also excluded. We present a generalization of previously discovered protocols for non-locality distillation \cite{BS2009} to the case of generalized PR boxes with more outcomes and parties, thus excluding new classes of boxes from the quantum set. Finally, the fact that vertices are decoupled from the environment (in particular, eavesdroppers in cryptographic tasks) leads us to consider if any box in the interior of the no-signaling polytope is always \textit{non-locally} correlated with the environment. We present an extension (into an $n+1$ party box) for the class of $n$ party boxes realized as admixtures of the generalized PR box with the isotropic completely noisy box, showing that indeed for any non-zero fraction of noise, such a box exhibits non-local correlations with the eavesdropper's system. We conclude with some open questions. For convenience, the proofs are deferred to the Supplemental Material \cite{suppmat}, with sketches of the proofs given in the main text.

{\it Extremal no-signaling correlations.}
Consider an $n$-party Bell experiment. We use the labels $(\mathcal{A}_i, \mathcal{X}_i)$ with $i \in [n]$ (with $[n] \defeq \{1,\dots,n\}$), where the sets $\mathcal{X}_i$ of size $m_i$ denote the respective inputs $x_i$ of the $n$ parties, while the sets $\mathcal{A}_i$ of size $k_i$ denote their respective outputs $a_i$. The number of inputs $m_i$ and outputs $k_i$ for each party is arbitrary but for ease of notation we will consider $m_i = m$ and $k_i = k$, $\forall i \in [n]$, whenever such a simplification does not affect the generality of the argument. A box $\mathcal{P}$ describes a set of conditional probability distributions $P(\textbf{a}| \textbf{x})$ with $\textbf{a} = \{a_1, \dots, a_n \} \in \mathcal{A}, \textbf{x} = \{x_1, \dots, x_n \} \in \mathcal{X}$ where $\mathcal{A} = \mathcal{A}_1 \times \dots \times \mathcal{A}_n$ and similarly $\mathcal{X} = \mathcal{X}_1 \times \dots \times \mathcal{X}_n$. The corresponding Bell scenario is denoted $\mathcal{B}(n, m, k)$. 

The box $\mathcal{P}$ is a valid no-signaling box (satisfying the no-signaling principle of relativity) for the Bell scenario if it satisfies:
\begin{itemize}
\item Non-negativity: $P(\textbf{a} | \textbf{x}) \geq 0 \; \; \forall \textbf{a}, \textbf{x}$; 
\item Normalization: $\sum_{\textbf{a}} P(\textbf{a} | \textbf{x}) = 1 \; \; \forall \textbf{x}$; and 
\item No-signaling:
\begin{eqnarray}
\label{no-signal}
\sum_{a_i} P(\textbf{a}|\textbf{x}^{(i)}) = \sum_{a_i} P(\textbf{a}|\textbf{x'}^{(i)}) \quad \forall \textbf{a}, \textbf{x}^{(i)}, \textbf{x'}^{(i)}, i, 
\end{eqnarray}
\end{itemize}
where $\textbf{x}^{(i)} = \{x_1, \dots, x_{i-1}, x_i, x_{i+1}, \dots, x_n\}$ and $\textbf{x'}^{(i)} = \{x_1, \dots, x_{i-1}, x'_i, x_{i+1}, \dots, x_n\}$.

The set of all boxes satisfying the above conditions forms the no-signaling convex polytope $\mathbf{NS}(n,m,k)$ which is written succinctly as $$\mathbf{NS}(n,m,k) = \{ \mathcal{P} \in \mathbb{R}^{(mk)^n} : \mathit{A} \cdot | \mathcal{P} \rangle \leq | \mathit{b} \rangle \}$$ and is of dimension $D = \prod_{i=1}^{n} \left[ m_i (k_i - 1) +1 \right] - 1$ (see for example \cite{Pironio}). Here the constraints of non-negativity, normalization and no-signaling are mathematically written in terms of the matrix $\mathit{A}$ and the vector $| \mathit{b} \rangle$, and the box $\mathcal{P}$ is written as a vector of length $(mk)^n$. 
Boxes that in addition satisfy the integrality constraint: 
\begin{itemize}
\item Integrality: $P(\textbf{a} | \textbf{x}) \in \{ 0, 1 \}$
\end{itemize}
are said to be \textit{classical} (deterministic) boxes $\mathcal{P}_d$. The convex hull of these deterministic boxes gives rise to the classical polytope $\mathbf{C}(n,m,k)$, i.e., the set of correlations obtainable from local hidden variable theories.

The set of quantum correlations $\mathbf{Q}(n,m,k)$ is defined as follows: $\mathcal{P} \in \mathbf{Q}(n,m,k)$ if there exist a state $\rho \in \mathcal{H}_d$ for some arbitrary dimension $d$, sets of measurement operators $\{E^{{x_i},{a_i}}_{i}\}$ for each party such that for all inputs and outputs 
\begin{equation}
P(\textbf{a} | \textbf{x}) = \Tr[\rho \otimes_{i=1}^{n} E^{{x_i},{a_i}}_{i}]
\end{equation}
with the measurement operators satisfying the requirements of hermiticity (${E^{{x_i},{a_i}}_{i}}^{\dagger} = E^{{x_i},{a_i}}_{i}, \; \; \forall x_i, a_i$), orthogonality ($E^{{x_i},{a_i}}_{i} E^{{x_i},{a'_i}}_{i} = \delta_{a_i, a'_i} E^{{x_i},{a_i}}_{i}, \; \; \forall x_i$) and completeness ($\sum_{a_i} E^{{x_i},{a_i}}_{i} = \identity, \;\; \forall x_i$).
This set is convex but in general is not a polytope, and we have the inclusions $\mathbf{C}(n,m,k) \subseteq \mathbf{Q}(n,m,k) \subseteq \mathbf{NS}(n,m,k)$. 

We are interested in the vertices of the no-signaling polytope, a box $\mathcal{P}$ being a vertex if it cannot be expressed as a non-trivial convex combination of the boxes in $\mathbf{NS}(n,m,k)$. Here, the following mathematical characterization of vertices will prove useful: \textit{every vertex satisfies in a unique way a certain number of the inequality constraints in $\mathit{A} \cdot | \mathcal{P} \rangle \leq |\mathit{b} \rangle$ with equality}, the corresponding submatrix $\mathit{\tilde{A}}(\mathcal{P})$ can therefore be used to uniquely identify the vertex (see Fact 1 in the Supplemental Material). 
A \textit{non-local vertex} is one that does not belong to $\mathbf{C}(n,m,k)$, in particular it does not contain only $0, 1$ entries. 

We use recently discovered connections between contextuality and non-local game scenarios to graph theory \cite{Winter, CSW,  Principles4, Principles5, Fritz} to show that such a non-local vertex does not belong to the quantum set $\mathbf{Q}(n,m,k)$ nor to its closure $\text{cl}(\mathbf{Q}(n,m,k))$. Our main result is then the following Theorem \ref{non-local-vertex}. For convenience, we sketch a proof here with the detailed proof deferred to the Supplemental Material.  
\begin{thm}
\label{non-local-vertex}
For some (arbitrary) $(n, m, k)$, let $\mathcal{P}^{nl}$ be a vertex of the no-signaling polytope $\mathbf{NS}(n,m,k)$ such that $\mathcal{P}^{nl} \notin \mathbf{C}(n,m,k)$. Then $\mathcal{P}^{nl} \notin \text{cl}(\mathbf{Q}(n,m,k))$. 
\end{thm}
\textit{Sketch of proof.} We first consider the hierarchy of outer semidefinite programming relaxations of the quantum set introduced in \cite{NPA}. In particular, one set of the hierarchy called the \textit{almost quantum} set $\mathbf{\tilde{Q}}(n,m,k)$ \cite{AQ} is identified following similar arguments to \cite{Fritz} with a convex set $TH^c(G)$, which is the set $TH(G)$ from graph theory with extra \textit{clique equalities} $C_{n,ns}$ added to take care of the normalization and no-signaling constraints. Here, $G$ is the orthogonality graph \cite{Winter} associated to the Bell scenario $\mathcal{B}(n,m,k)$, for a description of its construction as well as the associated notion of an orthogonal representation of a graph, see the detailed proof in the Supplemental Material. The set $TH^c(G) = \mathbf{\tilde{Q}}(n,m,k)$ is therefore defined as follows:
\begin{eqnarray}
TH^c(G) & \defeq &\{\mathcal{P} = (|\langle \psi |u_i \rangle|^2: i \in V) \in \mathbb{R}_{+}^{V}: \nonumber \\
&&\left \| \psi \right \| = \left \| |u_i \rangle \right\| = 1, \nonumber \\
&& \{|u_i \rangle\} \text{is an orthogonal representation of G}, \nonumber \\
&&\forall c \in C_{n, ns}, \; \; \sum_{i \in c} |\langle \psi | u_i \rangle|^2 = 1.\}
\end{eqnarray} 
The set $TH(G)$ was studied as a semidefinite programming relaxation to the stable set polytope $STAB(G)$ from graph theory in \cite{Lovasz-2}.

From Fact 1 in the Supplemental Material, the non-local vertices of $NS(n,m,k)$ 
must necessarily contain non-integral entries.
In particular, we know that to every non-local vertex $P$ there is an associated matrix $\tilde{A}(P)$ of rank equal to the dimension of the polytope $\prod_{i=1}^{n} \left[ m_i (k_i - 1) +1 \right] - 1$ such that $\tilde{A}(P) \cdot |\mathcal{P} \rangle = | \tilde{\mathit{b}} \rangle$. Now, the set $TH(G)$ is known to admit an alternative characterization in terms of an infinite number of inequalities \cite{Fujie-Tamura}, and one can show that every non-local vertex violates one of these inequalities. Furthermore, the minimum eigenvalue of the positive definite matrix $\tilde{M}(P) \defeq \tilde{A}(P)^T \tilde{A}(P)$ provides a lower bound on this violation. This argument is similar to arguments in \cite{Shepherd, Fujie-Tamura, Marcel-2}, where it was shown that the set $TH(G)$ does not reach the non-integral vertices of $QSTAB(G)$ which is the polytope arising from the \textit{linear programming relaxation} to $STAB(G)$.  
An additional argument then shows that $\mathbf{\tilde{Q}}(n,m,k)$ is a closed set, and since the hierarchy converges to the closure of the quantum set $\text{cl}(\mathbf{Q}(n,m,k))$, one can exclude non-local vertices from it. This ends the sketch of the proof.


{\it Non-local games with no quantum winning strategy.} 
An immediate application of Theorem \ref{non-local-vertex} is that if one is able to identify a non-local game that has a single unique winning no-signaling strategy (where by a winning strategy we mean one that achieves maximal value $1$), then such a strategy (being a vertex) cannot be realized in quantum theory. So that denoting by $\omega_q(\textsl{g})$ the quantum value of the game, 
we have $\omega_q(\textsl{g}) \neq 1$. 

Here, we consider a class of Bell inequalities known as total \textit{unique games} ($\textsl{g}^U$) for multiple players, which are a family of games of interest in the field of hardness of approximation (in determining the algorithmic complexity of finding close to optimal solutions for optimization problems) \cite{Khot}. A unique game is defined by the following winning condition: for each $\textbf{x}$ and each set of outcomes of any chosen $n-1$ parties, $\textbf{a}_{n-1}^{(j)} = \textbf{a} \setminus a_j \; \forall j$, the remaining party is required to output a single unique $a_j$ specified by a function $a_j = \sigma_{\textbf{x}}^{(j)}\left(\textbf{a}_{n-1}^{(j)} \right)$ with $\sigma_{\textbf{x}}^{(j)}\left(\textbf{a}_{n-1}^{(j)} \right) \neq \sigma_{\textbf{x}}^{(j)}\left(\textbf{a'}_{n-1}^{(j)}\right)$ for $\left(\textbf{a}_{n-1}^{(j)} \neq \textbf{a'}_{n-1}^{(j)}) \right)$; the term total refers to the fact that such a winning constraint is imposed for every set of inputs $\textbf{x}$. We give a simple condition in Lemma \ref{vert-conn-graph} for a total function unique game to be won by a single no-signaling box in terms of the connectivity of a certain graph that we associate to the game. 


\begin{lemma}
\label{vert-conn-graph}
A total multi-player unique game $\textsl{g}^{U}$ is won by a single unique non signaling box if and only if the no-signaling graph $G_{NS}(\textsl{g}^U)$ associated to the game is connected; for such a game $\omega_q(\textsl{g}^{U}) \neq 1$.
\end{lemma}
We illustrate the above application with an example of the generalized PR box and show that the corresponding generalized CHSH game cannot be won using a quantum strategy (the proofs of Lemma \ref{vert-conn-graph} and Example \ref{gen-pr-box} are given in the Supplemental Material).
\begin{example}
\label{gen-pr-box}
The unique game $\textsl{g}^{U, ex}$ for $n$ parties with inputs $x_i \in \{0,1\}$ and outputs $a_i \in \{0, \dots, k-1\}$ defined by the winning constraint: $\oplus_{i=1}^{n} a_i \; \text{(mod k)} = \prod_{i=1}^{n} x_i$ is won by a single no-signaling vertex for any $n, k$. 
\end{example}

{\it Exclusion of boxes by distilling many copies to a vertex.} The fact that non-local extremal points of the no-signaling set in a single-shot scenario cannot be realized using quantum resources can also be used to exclude from the quantum set those boxes which can be distilled (with many copies in a sequential scenario) to such a vertex. In other words, if there exists a distillation protocol that deterministically maps a certain number of copies of a box $\mathcal{P}$ to a non-local vertex, then $\mathcal{P} \notin \mathbf{Q}(n,m,k)$.

As an illustration, consider the generalized PR-box presented in the Example \ref{gen-pr-box}, defined as
\begin{equation}
\label{nl-vert-ex}
\mathcal{P}^{nl, ex}_{n}(\textbf{a} | \textbf{x}) =  \left\{ 
\begin{array}{ l r } 
\frac{1}{k^{n-1}} : \oplus_{i=1}^{n} a_i \; \text{(mod k)} = \prod_{i=1}^{n} x_i \\
0 \; \; \quad :  \text{otherwise} \\
 \end{array} 
 \right.
\end{equation}
We now examine an admixture of the above non-local vertex with a locally correlated box $\mathcal{P}^{c,ex}$ defined as 
\begin{equation}
\mathcal{P}^{c,ex}_{n}(\textbf{a} | \textbf{x}) =  \left\{ 
\begin{array}{ l r } 
\frac{1}{k^{n-1}} : \oplus_{i=1}^{n} a_i \; \text{(mod k)} = 0 \\
0 \; \; \quad :  \text{otherwise} \\
 \end{array} 
 \right.
\end{equation}
A generalization of the distillation protocols in \cite{BS2009}, \cite{EW2014}, gives that for any $\epsilon \in (0,1)$, many copies of the box 
\begin{equation}
\label{nl-mix}
\mathcal{P}_{\epsilon, n} = \epsilon \mathcal{P}^{nl,ex}_{n} + (1 - \epsilon) \mathcal{P}^{c,ex}_{n}
\end{equation}
can be distilled to the vertex $\mathcal{P}^{nl,ex}_{n}$, so that the box $\mathcal{P}_{\epsilon, n}$ for arbitrary $\epsilon (> 0), n$ (arbitrarily close to the classical polytope) is immediately excluded from the quantum set. 
\begin{lemma}
For any $\epsilon \in (0,1)$ there exists a deterministic distillation protocol $\mathbb{D}$ that maps $\mathcal{P}_{\epsilon, n}^{\otimes r}$ to $\mathcal{P}^{nl,ex}_{n}$ for large $r$, i.e., $\mathbb{D}: \mathcal{P}_{\epsilon, n}^{\otimes r} \xrightarrow{r \rightarrow \infty} \mathcal{P}^{nl,ex}_{n}$. 
\end{lemma}
While not all super-quantum boxes can be excluded by distillation to a vertex \cite{AQ}, the above example still motivates the study of non-locality distillation as a powerful tool for exclusion of super-quantum boxes in general Bell scenarios. 

{\it Extensions of non-extremal boxes.}
Perhaps the most crucial consequence of Theorem \ref{non-local-vertex} is that it rules out the most straightforward device-independent protocol for obtaining secure keys or random bits, as we now demonstrate. Extremal non-local boxes of the no-signaling polytope by definition cannot be correlated with the environment, i.e., if $\mathcal{P} = \{P(\textbf{a}|\textbf{x})\}$ is an $n$-party box that is a vertex, and $\{P(\textbf{a}, b|\textbf{x}, y)\}$ is any $n+1$-party box of which $\mathcal{P}$ is the marginal, then necessarily one has 
\begin{equation}
\{P(\textbf{a}, b|\textbf{x}, y)\} = \{P(\textbf{a} | \textbf{x})\} \otimes \{P(b|y)\}.
\end{equation}  
This was first noted in \cite{BLMPPR2005}, for a proof observe that
\begin{eqnarray}
\label{vertex-decorrelated}
P(\textbf{a}|\textbf{x}) = \sum_{b} P(\textbf{a}, b|\textbf{x}, y)  
= \sum_{b} P(\textbf{a}|\textbf{x}, b, y) P(b|y),
\end{eqnarray}
where we have used the no-signaling condition $P(b | \textbf{x}, y) = P(b|y)$. 
Since $\{P(\textbf{a}|\textbf{x})\}$ is a vertex, the decomposition in Eq.(\ref{vertex-decorrelated}) can only have a single term, so that for all $b, y$, we have $P(\textbf{a}|\textbf{x}, b, y) = P(\textbf{a} | \textbf{x})$, which translates to 
\begin{equation}
\{P(\textbf{a}, b|\textbf{x}, y)\} = \{P(\textbf{a} | \textbf{x}, b, y) P(b|y)\} = \{P(\textbf{a}|\textbf{x})\} \otimes \{P(b|y)\}.
\end{equation} 
Access to an extremal box $\mathcal{P}$ would thus provide a tremendous advantage to the honest parties in cryptographic tasks where one considers the presence of an all-powerful adversary limited only by the relativistic causality constraint imposed by the no-signaling principle. For instance, in a noiseless scenario, one can readily design a protocol where the honest parties check for the presence of the extremal box in a device-independent manner and then obtain the intrinsic randomness from the box or use it to distribute secure keys, knowing that they are uncorrelated with the eavesdropper's system. Theorem \ref{non-local-vertex} is thus a no-go result against this straightforward protocol. On the other hand, if $\mathcal{P}$ is not a vertex, then Eq.(\ref{vertex-decorrelated}) shows that there can always be (classical) correlations with the eavesdropper's system. 

A natural question then arises (first raised in \cite{BLMPPR2005}) whether any box $\mathcal{P}$ that is in the interior of the no-signaling polytope also allows for \textit{non-local} correlations with an eavesdropper, i.e., if there is an extension into a box $\bar{\mathcal{P}}_{n+1}(\textbf{a}, e|\textbf{x}, z)$ such that $\mathcal{P}_{n}(\textbf{a}|\textbf{x})$ is obtained as the marginal of $\bar{\mathcal{P}}_{n+1}(\textbf{a}, e|\textbf{x}, z)$ upon tracing out the eavesdropper's system, and $\bar{\mathcal{P}}_{n+1}(\textbf{a}, e|\textbf{x}, z)$ exhibits non-locality in the cut $(A_1, \dots, A_n) | E$. This is especially interesting not only due to the fact that non-local correlations impart additional power to the eavesdropper, but that the analogous fact in quantum theory holds, namely a mixed state of $n$ systems can always be extended into an \textit{entangled} state of $n+1$ systems, i.e., $\rho_{n} = \sum_{i} \lambda_i |i^{\otimes n} \rangle \langle i^{\otimes n} |$ can be extended to $|\psi_{n+1} \rangle = \sum_{i} \sqrt{\lambda_i} |i \rangle^{\otimes {n+1}}$ where the $|i \rangle$ constitute an orthonormal basis, and $\lambda_i > 0$ are the corresponding eigenvalues. For the generalized PR box from Eq.(\ref{nl-vert-ex})), which corresponds to the Bell scenario $\mathcal{B}(n,2,k)$, we show that any box obtained by admixing the vertex with the completely noisy box $\mathcal{P}^{\identity}_{n}$ with entries $\mathcal{P}^{\identity}_{n}(\textbf{a}|\textbf{x}) = \frac{1}{k^n} \; \forall \textbf{a}, \textbf{x} \;$ 
\textit{does} admit an extension into an $n+1$ party box $\bar{\mathcal{P}}_{n+1}(\textbf{a},e|\textbf{x},z)$ with $e \in \{0,\dots,k-1\}$ and $z \in \{0,1\}$ such that the bipartite cut $(A_1,\dots,A_n) \vert E$ exhibits non-local correlations, i.e., we now examine the isotropic boxes of the form  
\begin{equation}
\label{nl-mix-2}
\mathcal{P}^{\text{iso}}_{\epsilon,n} = \epsilon \mathcal{P}^{nl,ex}_{n} + (1-\epsilon) \mathcal{P}^{\identity}_{n}.
\end{equation} 
We note that the extremal points of the Bell scenario $\mathcal{B}(2,2,k)$ have been characterized in \cite{BLMPPR2005} where it was shown that all the extremal boxes are of the form given in Eq. (\ref{nl-vert-ex}) for $n=2$ up to a local relabeling of the inputs and outputs. 
 
\begin{lemma}
For the Bell scenario $\mathcal{B}(n,2,k)$, for any $\epsilon \in [0,1)$ the box $\mathcal{P}^{\text{iso}}_{\epsilon, n}$ admits a non-local extension to a box $\bar{\mathcal{P}}_{n+1}(\textbf{a},e|\textbf{x},z)$, with $a_i, e \in \{0,\dots,k-1\}, x_i ,z \in \{0,1\}$. 
\end{lemma} 
We leave as an open question whether other boxes in the interior of the no-signaling polytope also admit non-local extensions.

{\it Conclusions.} We have seen that the non-local vertices of the no-signaling polytope of correlations admit no quantum realization and explored some consequences of this fact. An interesting question for future study is how close the quantum set can get to such a vertex, and whether a universal finite lower bound on this distance can be found. Since the appearance of this paper, finite lower bounds on the distance in the $(2,2,k)$ and $(2,m,2)$ scenarios have been discovered \cite{Julio}. In this regard, it is pertinent to note that for Bell inequalities with algebraic violation (such as the GHZ paradoxes \cite{GHZ}) or close to algebraic violation (such as the chained Bell inequalities \cite{BC}), the quantum set only reaches the corresponding facet of the no-signaling polytope. Another question for further research is to identify other interesting classes of vertices than the one illustrated here, and construct novel distillation protocols for these. An open problem is whether all points in the interior of the no-signaling set allow for non-local correlations with the environment, and its consequences for cryptographic tasks. Finally, it would be important to extend the results in the paper to the scenario of sequential non-locality, especially from the point of view of device-independent applications. 

To address the question of finding a universal bound on the distance from the quantum set to extremal no-signaling boxes, one approach is to use the norm-based bounds developed in \cite{Julio, RRG}, a more general approach is to utilize the method of proof of the Theorem \ref{non-local-vertex} and lower bound the minimum eigenvalue of the matrix $\mathit{\tilde{A}}(\mathcal{P})^T \mathit{\tilde{A}}(\mathcal{P})$ associated with the vertex $\mathcal{P}$. Concerning non-locality distillation, recently in \cite{BG14} a systematic method to construct sets of non-local boxes that are closed under wirings has been introduced. In particular, a measure of correlations called the \textit{maximal correlation} has been shown to be monotonically decreasing under wirings. Limitations on the power of non-locality distillation to exclude super-quantum boxes may be studied by computing this measure on known classes of extremal boxes \cite{YCATS, BLMPPR2005} and their admixtures with local boxes. To address the question whether points in the interior of the no-signaling set allow for non-local correlations with the environment, it would be ideal to develop a notion of purification for no-signaling boxes, perhaps within the framework of generalized probabilistic theories \cite{Barrett}. 

{\em Acknowledgements.} This work is supported by the ERC grant QOLAPS, the EU grant RAQUEL (No. 323970), and the Foundation for Polish Science TEAM project co-financed by the EU European Regional Development Fund. R. R. acknowledges useful discussions with G. Murta, P. Mironowicz and M. Paw{\l}owski. 


{\bf Supplemental Material}

Here, we present the proofs of the lemmas and theorems stated in the main text.

The vertex of a polytope is a point such that the normal cone to the point has full dimension. Let us first state the following fact \cite{Schrijver}, which is used to characterize the extremal boxes (vertices) $\mathcal{P}$ of the convex no-signaling polytope.
\begin{fact}[\cite{Schrijver}, see also Theorem 2.5.3 in \cite{Fritz}]
\label{vertex-poly}
A box $\mathcal{P}$ is a vertex of the no-signaling polytope $\mathbf{NS}(n,m,k)$ if any only if $rank(\mathit{\tilde{A}}) = (mk)^n$ where $\mathit{\tilde{A}}$ denotes the sub-matrix of $\mathit{A}$ consisting of those row vectors $\mathit{A}_i$ for which $\mathit{A}_i \cdot | \mathcal{P} \rangle = | \mathit{b} \rangle_i$.
\end{fact}
For two distinct vertices $\mathcal{P}$ and $\mathcal{P'}$, the corresponding sub-matrices are not equal $\left(\mathit{\tilde{A}}(\mathcal{P}) \neq \mathit{\tilde{A}}(\mathcal{P'})\right)$ so that the matrix $\mathit{\tilde{A}}(\mathcal{P})$ can be used to identify the vertex $\mathcal{P}$. 

We also introduce the notion of a \textit{no-signaling graph} $G_{NS}(\textsl{g})$ that can be associated with any non-local game $\textsl{g}$.  
\begin{mydef}
For any non-local game $\textsl{g}$, we define the no-signaling graph $G_{NS}(\textsl{g}) = (\textsl{V}, \textsl{E})$ associated with the game to have set of vertices $\textsl{v} \in \textsl{V} $, each of which is labeled by a set of inputs and outputs that wins the game, $\textsl{v}= \left(\textbf{a}^{(\textsl{v})}, \textbf{x}^{(\textsl{v})} \right)$. Two vertices $\textsl{v}, \textsl{v'} \in \textsl{V}$ are connected be an edge if $\exists S \subseteq [n] \; with \; \left|S\right| = n-1$ such that $(a^{(\textsl{v})}_i=a^{(\textsl{v'})}_i \; \wedge \;  x^{(\textsl{v})}_i=x^{(\textsl{v'})}_i)\; \forall i \in S$. 
\end{mydef}

\textit{\textbf{Lemma 2.}
A total multi-player unique game $\textsl{g}^{U}$ is won by a single unique non signaling box if and only if the no-signaling graph $G_{NS}(\textsl{g}^U)$ associated to the game is connected; for such a game $\omega_q(\textsl{g}^{U}) \neq 1$.}
\begin{proof}
By definition for a unique game, only one term from the game constraint appears in each sum in the no-signaling constraint
\begin{eqnarray}
\label{no-signal}
\sum_{a_i} P(\textbf{a}|\textbf{x}^{(i)}) = \sum_{a_i} P(\textbf{a}|\textbf{x'}^{(i)}) \qquad \forall \textbf{a}, \textbf{x}^{(i)}, \textbf{x'}^{(i)}, i, 
\end{eqnarray}
for any $\textbf{x}^{(i)}, \textbf{x'}^{(i)}, \textbf{a}$. Hence, two vertices being connected implies that the corresponding probabilities are equal. Consequently, the no-signaling graph $G_{NS}(\textsl{g}^U)$ being connected implies that the box is uniform, with every winning event occurring with probability $\frac{1}{k}$, where $k$ is the number of outputs of the game. Since the total unique game considers all sets of inputs $\textbf{x}$, it is won by this specific no-signaling box alone, which must necessarily therefore be a vertex of the no-signaling polytope. By Theorem 1, we know that such a box cannot be realized by measurements on quantum states, so that $\omega_q(\textsl{g}^U) \neq 1$. 

On the other hand, consider that the graph $G_{NS}(\textsl{g}^{U})$ is disconnected, being composed of $c > 1$ components $G_c$. By definition of the no-signaling graph and the winning constraint of the unique game, each $G_c$ consists of the same number $k_c$ of vertices from each setting $\textbf{x}$, i.e. $k_c(\textbf{x}) = k_c(\textbf{x'})$ for any pair of $\textbf{x}, \textbf{x'}$, where $k_c(\textbf{x})$ denotes the number of vertices $\textsl{v}$ in graph $G_c$ with associated label $\left(\textbf{a}^{(\textsl{v})}, \textbf{x} \right)$. Corresponding to each $G_c$ therefore, we can construct the explicitly non signaling box $\mathcal{P}_c$ with entries $\mathcal{P}_c(\textbf{a}^{\textsl{v}} | \textbf{x}^{\textsl{v}}) = \frac{1}{k_c}$ if the vertex $(\textbf{a}^{\textsl{v}}, \textbf{x}^{\textsl{v}}) \in \textsl{V}(G_c)$, and zero otherwise. It is readily seen that each such box $\mathcal{P}_c$ wins the game.
\end{proof}
An alternative (equivalent) characterization of the total unique game with single maximally violating box is that the block matrix composed of the permutation matrices associated to the functions $\sigma_{\textbf{x}}^{(j)}$ defining the unique game in every $n-1 \vert 1$ cut is irreducible.

\textit{\textbf{Example 1.}
The unique game $\textsl{g}^{U, ex}$ for $n$ parties with inputs $x_i \in \{0,1\}$ and outputs $a_i \in \{0, \dots, k-1\}$ defined by the winning constraint: $\oplus_{i=1}^{n} a_i \; \text{(mod k)} = \prod_{i=1}^{n} x_i$ is won by a single no-signaling vertex for any $n, k$. }

\begin{proof}
We show that the no-signaling graph $G_{NS}(\textsl{g}^{U, ex})$ is connected so that the game is won by a single no-signaling box by Lemma 2. Firstly, note that the no-signaling graph $G_{NS}(\textsl{g}^U)$ for any unique game of $n$ parties with $m$ inputs each with $k$ outputs is of size $m^n \cdot k^{n-1}$ since for each set of inputs $\textbf{x}$ (of which there are $m^n$) there are $k^{n-1}$ outputs $\textbf{a}$ that win the game. Therefore, for $G_{NS}(\textsl{g}^{U,ex})$ with vertex set $\textsl{V}$, we have $|\textsl{V}| = 2^n k^{n-1}$. We may group the vertices into two sets $\textsl{V}_j \defeq \left\lbrace \textsl{v} : \oplus_{i=1}^{n} a_i^{(\textsl{v})} \; \text{(mod k)} = \prod_{i=1}^{n} x_i^{(\textsl{v})} = j \right\rbrace$ for $j \in \{0, 1\}$. 

The edges of the graph $G_{NS}(\textsl{g}^{U, ex})$ can also be grouped into two sets based on $\textsl{V}_0$ and $\textsl{V}_1$:  
\begin{enumerate}
\item For $\textsl{u}, \textsl{v} \in \textsl{V}_0$, we have the edge $\textsl{u} \sim \textsl{v}$ if $\left(\textbf{a}^{(\textsl{v})} = \textbf{a}^{(\textsl{u})}\right) \wedge \left(\textbf{x}^{(\textsl{v})} = \textbf{x}^{(\textsl{u})}_{\oplus_2 i}\right)$ for some $i \in [n]$ where $\textbf{x}^{(\textsl{u})}_{\oplus_2 i} = \{x_1^{(\textsl{u})}, \dots,x_{i-1}^{(\textsl{u})}, x_i^{(\textsl{u})} \oplus_2 1, x_{i+1}^{(\textsl{u})}, \dots, x_n^{(\textsl{u})}\}$, and $\oplus_j$ denotes addition modulo $j$. 
\item For $\textsl{u} \in \textsl{V}_1, \textsl{v} \in \textsl{V}_0$, we have $\textsl{u} \sim \textsl{v}$ if $\left(\textbf{a}^{(\textsl{v})} = \textbf{a}^{(\textsl{u})}_{\oplus_k i} \right) \wedge \left(\textbf{x}^{(\textsl{v})} = \textbf{x}^{(\textsl{u})}_{\oplus_2 i} \right)$ for some $i \in [n]$.  
\end{enumerate}
Therefore, we have that the graph is regular with the degree of each vertex (the number of edges incident at each vertex) equal to $n$, so that the total number of edges in the graph is given by $|\textsl{E}| = \frac{n 2^n k^{n-1}}{2}$.
Let us divide the vertex set $\textsl{V}$ into $k^{n-1}$ subsets each of fixed $\textbf{a}_{n-1}$ (outputs of the first $n-1$ parties) which we denote by $\textsl{V}^{\textbf{a}_{n-1}}$ with $|\textsl{V}^{\textbf{a}_{n-1}}| = 2^n$. From the edge sets above, it is clear that the corresponding subgraphs $G^{\textbf{a}_{n-1}}$ are connected. It remains to be shown that these $k^{n-1}$ subgraphs are connected to each other. 

To do this, let us form a new graph $G'_{NS}$ with each connected subgraph $G^{\textbf{a}_{n-1}}$ as a single vertex in the new graph and edges between two vertices in $G'_{NS}$ if there exists an edge between any two vertices in the corresponding subgraphs. Due to the $\textsl{xor}$ structure, the edges in set (2) above with $\textsl{u} \in \textsl{V}_1$ and $\textsl{v} \in \textsl{V}_0$ give rise to cyclic subgraphs of length $k$ in $G'_{NS}$. The vertex set of $G'_{NS}$ (of size $k^{n-1}$) is thus partitioned in $n-1$ different ways into $k^{n-2}$ sets of vertices belonging to disjoint cyclic subgraphs of size $k$ each. This guarantees that the graph $G'_{NS}$ is connected; and hence that the different subgraphs $G^{\textbf{a}_{n-1}}$ for each $\textbf{a}_{n-1}$ are connected to each other. The graph $G_{NS}(\textsl{g}^{U, ex})$ is thus connected; by Lemma 2, we have that the game $\textsl{g}^{U, ex}$ is won by a unique no-signaling vertex, and hence cannot be won using quantum resources.
\end{proof}
In \cite{Cleve}, it was shown that for two-player binary games $\textsl{g}^{\text{bin}}$ (where the output alphabet of the two players is binary), $\omega_q(\textsl{g}^{\text{bin}}) = 1$ if and only if $\omega_c(\textsl{g}^{\text{bin}}) = 1$. For the sub-class of binary games known as (total) \textsc{xor} games (where the game winning constraint only depends on the \textsc{xor} of the two players' outputs), this result can be seen as a consequence of Lemma 2. This is because for total $\textsc{xor}$ games, it is readily seen that the graph $G_{NS}$ is disconnected (and the corresponding block matrix composed of the two possible permutation matrices corresponding to correlations and anti-correlations is reducible) if and only if the game can be won with a classical deterministic strategy.

\textit{\textbf{Lemma 3.}
For any $\epsilon \in (0,1)$ there exists a deterministic distillation protocol $\mathbb{D}$ that maps $\mathcal{P}_{\epsilon, n}^{\otimes r}$ to $\mathcal{P}^{nl,ex}_{n}$ for large $r$, i.e., $\mathbb{D}: \mathcal{P}_{\epsilon, n}^{\otimes r} \xrightarrow{r \rightarrow \infty} \mathcal{P}^{nl,ex}_{n}$. }
\begin{proof}
We show a protocol $\mathbb{D}^{(2)}$ that maps $\mathcal{P}_{\epsilon, n}^{\otimes 2}$ to $\mathcal{P}_{\epsilon', n}$ with $\epsilon' > \epsilon$; the statement then follows by iteration of $\mathbb{D}^{(2)}$. Let us denote by $x_{i}^{(j)}$ and  $a_{i}^{(j)}$ the input and outputs of the $i$-th party to the $j$-th copy of the box, where $i \in [n]$ and $j \in \{1,2\}$. The protocol $\mathbb{D}^{(2)}$ maps to the box with inputs $x_{i}$ and outputs $a_{i}$. The inputs to the two boxes are given by $x_{i}^{(1)} = x_{i}$ and $x_{i}^{(2)} = x_{i} \cdot \mathsf{I}(a_{i}^{(1)} \oplus_k 1)$, where $\mathsf{I}$ is the indicator $\mathsf{I}(b) = 1$ if $b = 1$ and $0$ otherwise. The $i$-th party finally outputs $a_{i} = a_{i}^{(1)} \oplus a_{i}^{(2)}$. We show that with probability $\epsilon' > \epsilon$, we have $\oplus_{i=1}^{n} a_i = \prod_{i=1}^{n} x_i$ so that the iteration of the protocol succeeds in distilling the box $\mathcal{P}_{\epsilon, n}$ to the vertex $\mathcal{P}^{nl,ex}_{n}$. Consider the action of the protocol on each of the four terms in $\mathcal{P}_{\epsilon, n}^{\otimes 2}$. (i) ${\mathcal{P}^{nl,ex}_{n}}^{\otimes 2}$: With the above choice of inputs, we have that $\oplus_{i=1}^{n} a_{i}^{(1)} = \prod_{i=1}^{n} x_i$ and $\oplus_{i=1}^{n} a_{i}^{(2)} = \prod_{i=1}^{n} x_i \cdot \mathsf{I}(a_{i} \oplus_k 1)$. It is readily seen that when $x_i = 0$ for some $i$, we have $\oplus_{i} a_i = \oplus_{i} \left(a_{i}^{(1)} \oplus a_{i}^{(2)} \right) = 0$; and when $x_i = 1 \; \forall i$, since the first box has $\oplus_{i} a_{i}^{(1)} = 1$, the function $\mathsf{I}(a_{i}^{(1)} \oplus_k 1)$ ensures that $\oplus_{i} a_{i}^{(2)} = 0$ so that once more we have $\oplus_{i} a_i = \prod_i x_i$. Therefore, the protocol maps the box ${\mathcal{P}^{nl,ex}_{n}}^{\otimes 2}$ to $\mathcal{P}^{nl,ex}_{n}$. (ii) $\mathcal{P}^{nl,ex}_{n} \otimes \mathcal{P}^{c,ex}_{n}$: by definition, we have $\oplus_{i} a_{i}^{(1)} = \prod_{i} x_i$ and $\oplus_{i} a_{i}^{(2)} = 0$, so that this box is also mapped to $\mathcal{P}^{nl,ex}_{n}$. (iii) $\mathcal{P}^{c,ex}_{n} \otimes \mathcal{P}^{nl,ex}_{n}$: we have that $\oplus_{i} a_{i}^{(2)} = \prod_{i} x_i$ when $a_i^{(1)} = 0 \; \forall i$ which happens with probability $\frac{1}{k^{n-1}}$ so that this box is mapped to $\frac{1}{k^{n-1}} \mathcal{P}^{nl,ex}_{n} + \left( 1 - \frac{1}{k^{n-1}} \right) \mathcal{P}^{c,ex}_{n}$. (iv) ${\mathcal{P}^{c,ex}_{n}}^{\otimes 2}$: we have that simply $\oplus_{i} a_i = 0$ so that this box is mapped to $\mathcal{P}^{c,ex}_{n}$. Overall, we see that the protocol $\mathbb{D}^{(2)}$ maps $\mathcal{P}_{\epsilon, n}^{\otimes 2}$ to $\epsilon' \mathcal{P}^{nl,ex}_{n} + \left(1 - \epsilon' \right) \mathcal{P}^{c,ex}_{n}$ with $\epsilon' = \frac{\epsilon}{k^{n-1}} \left(k^{n-1} + 1 - \epsilon \right) >  \epsilon$. Iterating this protocol for $r = 2^j$ boxes for large $j$ results in the distillation of the initial box $\mathcal{P}_{\epsilon, n}$ to the non-local vertex $\mathcal{P}^{nl,ex}_{n}$.  
\end{proof}

\textit{\textbf{Lemma 4.}
For the Bell scenario $\mathcal{B}(n,2,k)$, for any $\epsilon \in [0,1)$ the box $\mathcal{P}^{\text{iso}}_{\epsilon}$ admits a non-local extension to a box $\bar{\mathcal{P}}(\textbf{a},e|\textbf{x},z)$, with $a_i, e \in \{0,\dots,k-1\}, x_i ,z \in \{0,1\}$.}

\begin{proof}
Define the non-local box $\bar{\mathcal{P}}_{\epsilon, n+1}$ as
\begin{equation}
\bar{\mathcal{P}}_{\epsilon, n+1} \defeq \epsilon \mathcal{P}^{nl,ex}_{n} \otimes \mathcal{P}^{d}_{1} + (1 - \epsilon) \tilde{\mathcal{P}}^{nl,ex}_{n+1},
\end{equation} 
where $\mathcal{P}^{d}_{1}$ is a deterministic box with entries $\mathcal{P}^{d}_{1}(0|z) = 1 \; \; \forall z$ and $0$ otherwise, and $\tilde{\mathcal{P}}^{nl,ex}_{n+1}$ is a slightly modified form of the generalized PR box, obeying $\oplus_{i=1}^{n+1} a_i = x_1 \cdot  (x_{2} \oplus_2 1) \cdot \dots \cdot (x_n \oplus_2 1) \cdot x_{n+1}$. It is then readily seen that $\mathcal{P}_{\epsilon, n}$ is obtained from $\bar{\mathcal{P}}_{\epsilon, n+1}$ by tracing out the $(n+1)$-th system. 

We now show that for any $\epsilon \in [0,1)$, the box $\bar{\mathcal{P}}_{\epsilon,n+1}$ in the cut $(A_1, \dots, A_n) \vert E$ violates the bipartite Bell inequality defined for $\mathcal{P}^{nl,ex}_{2}$, i.e., with the winning constraint $a \oplus e= x \cdot z$. This is achieved by the following choice of inputs for the $n+1$ parties: $(x_1 = x, x_2 = \dots x_n = 0, x_{n+1} = z)$ and the outputs $(a = \oplus_{i=1}^{n} a_i, e = a_{n+1})$. With this choice of inputs and outputs, the box $\bar{\mathcal{P}}_{\epsilon,n+1}$ achieves in the cut $(A_1, \dots, A_n) \vert E$ the value $4 - \epsilon$ for the bipartite inequality which beats the local bound of $3$ for any $\epsilon \in [0,1)$. 
\end{proof}

We now present the proof of the main Theorem 1 stated in the main text. To do this, we first present some background information on recently discovered relations between contextuality and non-local game scenarios to graph theory.

In \cite{NPA}, a hierarchy of semi-definite programs was formulated for optimization with non-commuting variables, that is used to identify when a box $\mathcal{P}$ does not belong to the quantum set; it was also shown that the hierarchy converges to a set $\mathbf{Q}^{pr}(n,m,k)$. Note that the set $\mathbf{Q}^{pr}(n,m,k)$ described in this hierarchy can be in fact larger than $\mathbf{Q}(n,m,k)$ since the former set is defined for the probabilities given by
\begin{equation}
P(\textbf{a}| \textbf{x}) = \langle \psi |\prod_{i=1}^{n} E^{{x_i},{a_i}}_{i} | \psi \rangle
\end{equation}
with commutativity imposed on the measurements of different parties, i.e., ($\left[E^{{x_i},{a_i}}_{i}, E^{{x_j},{a_j}}_{j} \right] = 0$ for $i \neq j$). In other words, the strict requirement of a tensor product structure of the operators of different parties is replaced with the requirement of only commutation between these operators. Since $\mathbf{Q}(n,m,k) \subseteq \mathbf{Q}^{pr}(n,m,k)$, the exclusion of the non-local vertices from $\mathbf{Q}^{pr}(n,m,k)$ is sufficient to exclude them from the quantum set. 

In the hierarchy, one considers sets of sequences of products of projection operators $S_1 = \{\identity\} \bigcup \{E^{x_{1}, {a_{1}}}_{1} \} \bigcup \dots \bigcup \{E^{x_{n}, {a_{n}}}_{n}\}$, $S_2 = S_1 \bigcup \{E^{{x_{i_1}},{a_{i_1}}}_{i_1} E^{{x_{i_2}},{a_{i_2}}}_{i_2}: i_1, i_2 \in [n]\}$, etc.
The convex sets $\mathbf{Q}_l$ corresponding to different levels of the hierarchy are then constructed by testing for the existence of a certificate $\Gamma^l$ associated to the set of operators $S_l$ using a semi-definite program. The certificate $\Gamma^l$ corresponding to level $l$ is a $|S_l| \times |S_l|$ matrix whose rows and columns are indexed by the operators in set $S_l$, and is a complex Hermitian positive semi-definite matrix that satisfies the following constraints on its entries
\begin{widetext}
\begin{eqnarray}
\Gamma^l_{\identity, \identity} = 1, \; \; \Gamma^l_{\identity, E^{{x_{i_1}},{a_{i_1}}}_{i_1} \dots E^{{x_{i_j}},{a_{i_j}}}_{i_j}} =  \Gamma^l_{E^{{x_{i_1}},{a_{i_1}}}_{i_1} \dots E^{{x_{i_j}},{a_{i_j}}}_{i_j},\identity}  &&= \Gamma^l_{E^{{x_{i_1}},{a_{i_1}}}_{i_1} \dots E^{{x_{i_j}},{a_{i_j}}}_{i_j}, E^{{x_{i_1}},{a_{i_1}}}_{i_1} \dots E^{{x_{i_j}},{a_{i_j}}}_{i_j}} = P(a_{i_1} \dots a_{i_j}|x_{i_1} \dots x_{i_j})  \nonumber \\
&& \qquad \qquad \qquad \qquad \qquad \qquad \qquad \qquad  (j \leq l) \nonumber
\end{eqnarray}
\end{widetext}
as well as constraints resulting from operator relations, $\Gamma^l_{Q, R} = \Gamma^l_{S, T}$ if $Q^{\dagger} R = S^{\dagger} T$. The lack of existence of a positive semi-definite certificate corresponding to any level of the hierarchy implies the exclusion of the box from the quantum set. Recently, one of the levels of the hierarchy labeled $\mathbf{\tilde{Q}}(n,m,k)$ has been highlighted as the \textit{almost quantum} set \cite{AQ}, owing to the fact that it was proven to satisfy most (all but the principle of \textit{information causality} \cite{Principles2}, there being numerical evidence in support of this too) of the reasonable information-theoretic principles designed to pick out quantum theory from among all no-signaling theories. 
The set $\mathbf{\tilde{Q}}$ is an intermediate level of the hierarchy associated to the set of operators $\tilde{S} = \identity \bigcup \{ E^{{x_{1}},{a_{1}}}_{1} \dots E^{{x_{n}},{a_{n}}}_{n}: a_i \in \mathcal{A}_i, x_i \in \mathcal{X}_i \}$.
As we shall see, $\mathbf{\tilde{Q}}(n,m,k)$ already excludes all the non-local no-signaling vertices. 


Recently, a relationship was discovered between Bell scenarios, more general contextuality scenarios, and graph theory via the use of orthogonality graphs \cite{Winter, Principles4} and hypergraphs \cite{Fritz}. The orthogonality graph $G_{\mathcal{B}}$ corresponding to  the Bell scenario $\mathcal{B}(n,m,k)$ is constructed as follows. Each event $(\textbf{a}^{\textsl{v}} | \textbf{x}^{\textsl{v}})$ of the Bell scenario corresponds to a distinct vertex $\textsl{v}$ of the graph and two such vertices are connected by an edge in the graph if the corresponding events are locally orthogonal \cite{Principles4}, i.e. 
\begin{equation}
\textsl{u} \sim \textsl{v} \Longleftrightarrow \; \; \exists i \in [n] \; \; s.t \; \; x_i^{\textsl{u}} = x_i^{\textsl{v}} \wedge a_i^{\textsl{u}} \neq a_i^{\textsl{v}}.
\end{equation} 
Equivalently, each operator of the form $E^{{x_{1}},{a_{1}}}_{1} \dots E^{{x_{n}},{a_{n}}}_n$ in the $n$-party Bell scenario corresponds to a vertex of $G_{\mathcal{B}}$, and a pair of vertices is connected by an edge if 
$\exists i \in [n]$ such that $E^{{x_{i}},{a_{i}}}_{i} E^{{x_{i}},{a'_{i}}}_{i} = 0$. The graph $G_{\mathit{B}}$ therefore has $\bar{D} \defeq (mk)^n$ vertices (in general the number of vertices is $\prod_{i=1}^{n} m_i k_i$). An orthonormal representation of a graph $G$ with vertex set $V$ and edge set $E$ is a set of unit vectors $\{| u_i \rangle: i \in V\}$ where $|u_i \rangle \in \mathbb{R}^N$ for some $N$, $\| |u_i \rangle \| = 1$ for all $i \in V$ and $\langle u_i | u_j \rangle = 0$ for $(i,j) \in E$ (note that this differs from the original representation in \cite{Lovasz-2} where \textit{non-adjacent} vertices were assigned orthogonal vectors). Based on orthonormal representations, one constructs the theta body TH(G) of the graph which is a convex set introduced in \cite{Lovasz-2} in relation to the graph-theoretic problem of finding a maximum weight stable set (a stable set is a set of mutually disconnected vertices) of the graph. We refer to \cite{Lovasz-2} for details. 
\begin{mydef}[\cite{Lovasz-2}]
For any graph $G = (V, E)$, $TH(G) \defeq \{\mathcal{P} = (|\langle \psi |u_i \rangle|^2: i \in V) \in \mathbb{R}_{+}^{V}: \left \| \psi \right \| = \left \| |u_i \rangle \right\| = 1, \{|u_i \rangle\}$ is an orthonormal representation of $G$.\}
\end{mydef}
We are now ready to provide the proof of the main theorem. 

\begin{thm}
For some (arbitrary) $(n, m, k)$, let $\mathcal{P}^{nl}$ be a vertex of the no-signaling polytope $\mathbf{NS}(n,m,k)$ such that $\mathcal{P}^{nl} \notin \mathbf{C}(n,m,k)$. Then $\mathcal{P}^{nl} \notin \text{cl}(\mathbf{Q}(n,m,k))$.
\end{thm}

\begin{proof}
We first examine the precise relation between $\mathbf{\tilde{Q}}$ and the set TH(G), noting that a relation in terms of (non-orthogonality) hypergraphs was proven in \cite{Fritz}. 
To do this, we first note that as shown in \cite{Principles4}, normalization and no-signaling constraints on a box are equivalent to (maximum) clique inequalities on the orthogonality graph, where a clique inequality is an inequality of the form $\sum_{i \in c} | \mathcal{P} \rangle_i \leq 1$ for some clique $c$ in the graph (a clique is a set of mutually connected vertices). The fact that normalization is a clique inequality that is saturated for the orthogonality graph is clear. The fact that no-signaling can also be formulated as a maximum clique inequality (that is also saturated) is seen as follows \cite{Principles4}. 
Using the normalization constraints, the no-signaling constraints can be rewritten as 
\begin{widetext}
\begin{eqnarray}
\sum_{a_j} P(a_1, \dots, a_j, \dots, a_n | x_1, \dots, x_j, \dots, x_n)  + \sum_{a'_j} \sum_{\substack{\{\tilde{a}_1, \dots \tilde{a}_{j-1}, \tilde{a}_{j+1} \dots \tilde{a}_n\} \neq \\ \{a_1, \dots, a_{j-1}, a_{j+1}, \dots, a_n\}}}  P(\tilde{a}_1, \dots, a'_j, \dots, \tilde{a}_n | x_1, \dots, x'_j, \dots, x_n) = 1 \nonumber \\
\qquad \qquad \qquad \qquad \qquad \; \; \forall j, x_j, x'_j, x_i \; (i \neq j)
\end{eqnarray}
\end{widetext}
Now, it is readily seen that all the events occurring in the above expression are locally orthogonal, so that the no-signaling constraints are maximum clique inequalities that are saturated for the orthogonality graph $G_{\mathit{B}}$; however we also note that not every maximum clique inequality of $G_{\mathit{B}}$ corresponds to a no-signaling or a normalization condition. We now define the subset $TH^{(c)}(G) \subset TH(G)$ as the set of boxes in TH(G) for which the maximum cliques $c$ corresponding to the normalization and no-signaling constraints have $\sum_{i \in c} |\langle \psi |u_i \rangle|^2 = 1$, i.e., defining $C_{n,ns}$ as the set of maximum cliques corresponding to the normalization and no-signaling constraints, we have the following definition for $TH^{(c)}(G)$. 
\begin{mydef}[see also the definition of $\mathcal{E}_{QM}^{1}(\Gamma)$ in \cite{Winter}]
For the orthogonality graph $G_{\mathit{B}}$ corresponding to the Bell scenario $\mathit{B}(n,m,k)$, define the set $TH^{(c)}(G)$ as $TH^{(c)}(G) \defeq \{\mathcal{P} = (|\langle \psi |u_i \rangle|^2) \in TH(G): \forall c \in C_{n, ns}, \; \; \sum_{i \in c} |\langle \psi | u_i \rangle|^2 = 1\}$.
\end{mydef}
\begin{thm}[see also Theorem 6.3.1 in \cite{Fritz}]
\label{theta-alq}
For any $n$-party Bell scenario $\mathit{B}$, $\mathbf{\tilde{Q}} = TH^{(c)}(G_{\mathit{B}})$.
\end{thm}


\begin{proof}

The set $\mathbf{\tilde{Q}}$ of boxes $\mathcal{P} = \{P(\textbf{a}| \textbf{x}) \}$ is characterized by the existence of a certificate $\Gamma^{\tilde{Q}}$ which is a $\bar{D} \times \bar{D}$ complex Hermitian positive semi-definite matrix with the following constraints on its entries
\begin{widetext}
\begin{eqnarray}
&&\Gamma^{\tilde{Q}}_{\identity, \identity} = 1, \nonumber \\ &&\Gamma^{\tilde{Q}}_{\identity, E^{{x_{1}},{a_{1}}}_{1} \dots E^{{x_{n}},{a_{n}}}_{n}} =  \Gamma^{\tilde{Q}}_{E^{{x_{1}},{a_{1}}}_{1} \dots E^{{x_{n}},{a_{n}}}_{n},\identity} = \Gamma^{\tilde{Q}}_{E^{{x_{1}},{a_{1}}}_{1} \dots E^{{x_{n}},{a_{n}}}_{n}, E^{{x_{1}},{a_{1}}}_{1} \dots E^{{x_{n}},{a_{n}}}_{n}} = P(\textbf{a}| \textbf{x}) \; \; \; \forall (\textbf{a}, \textbf{x}) \nonumber 
\end{eqnarray}
\end{widetext}
%
as well as constraints resulting from any operator relations $\Gamma^{\tilde{Q}}_{L, M} = \Gamma^{\tilde{Q}}_{N, O}$ if $L^{\dagger} M = N^{\dagger} O$. Crucially, the only constraints from this set of operators are those that impose the Hermiticity of the matrix by the identity
\begin{eqnarray}
&&(E^{{x_{1}},{a_{1}}}_{1} \dots E^{{x_{n}},{a_{n}}}_{n}) (E^{{x'_{1}},{a'_{1}}}_{1} \dots E^{{x'_{n}},{a'_{n}}}_{n})^{\dagger} \equiv \nonumber \\ && \qquad \qquad  \{(E^{{x'_{1}},{a'_{1}}}_{1} \dots E^{{x'_{n}},{a'_{n}}}_{n}) (E^{{x_{1}},{a_{1}}}_{1} \dots E^{{x_{n}},{a_{n}}}_{n})^{\dagger}\}^{\dagger} \nonumber
\end{eqnarray}
due to the commutativity of the operators $E^{{x_{i_l}},{a_{i_l}}}_{i_1}$ and $E^{{x_{i_2}},{a_{i_2}}}_{i_2}$ acting on the systems of different parties $i_1$ and $i_2$. As noted in \cite{NPA}, the question of existence of the complex Hermitian positive semi-definite matrix $\Gamma^{\tilde{Q}}$ is equivalent to the existence of a real symmetric positive semi-definite matrix defined as $\frac{1}{2}(\Gamma^{\tilde{Q}} + {\Gamma^{\tilde{Q}}}^*)$. The set $\mathbf{\tilde{Q}}$ then consists of those boxes for which the certificate $\Gamma^{\tilde{Q}}$ exists, and the normalization and no-signaling constraints are satisfied.

For a graph $G = (V, E)$, the theta body TH(G) can be equivalently defined as \cite{Fujie-Tamura}
\begin{widetext}
\begin{equation}
\label{theta-body-2}
TH(G) = \left\{ 
\begin{tabular}{ c|c c } 
$| \mathcal{P} \rangle \in \mathbb{R}^{|V|}$ & $\exists \mathbf{\Pi} \in \mathbb{S}^{|V|} \; \; \text{s.t.}$ & \;\;\;\; $\mathbf{\Pi}_{i,j} = 0 \quad ((i,j) \in E),$ \\
   &&$\mathbf{\Pi}_{i,i} = | \mathcal{P} \rangle_i \quad (i \in V),$ \\
   &&$\mathbf{\Pi} - | \mathcal{P} \rangle \langle \mathcal{P} | \succeq 0$\\
 \end{tabular} 
 \right\}
\end{equation}
\end{widetext}
where $\mathbb{S}^{|V|}$ denotes the set of symmetric matrices of size $|V| \times |V|$ and we view the box $\mathcal{P}$ as a vector in $\mathbb{R}^{|V|}$ (for our graphs, we have $|V| = \bar{D} = (mk)^n$). While there are other equivalent definitions of this set, the above definition is most suitable to recognize the similarities between this set and $\mathbf{\tilde{Q}}$. 
By inspecting the characterization of $TH^{(c)}(G)$ above, it can be seen that the set $TH^{(c)}(G_{\mathit{B}})$ of the orthogonality graph corresponding to the Bell scenario $\mathit{B}$ is just the set $\mathbf{\tilde{Q}}$ of the hierarchy. 
This is seen by considering the requirement of the real symmetric positive semi-definite matrix $(\frac{1}{2})(\Gamma^{\tilde{Q}} + {\Gamma^{\tilde{Q}}}^*)$ in $\mathbf{\tilde{Q}}$ and identifying it with the requirement of the real symmetric positive semi-definite matrix $\Bigl(\begin{smallmatrix}
\identity& \langle \mathcal{P} | \\ |\mathcal{P} \rangle & \mathbf{\Pi}
\end{smallmatrix} \Bigr)$ which by the use of Schur complements \cite{Horn} is equivalent to the condition that $\mathbf{\Pi} -  | \mathcal{P} \rangle \langle \mathcal{P} | \succeq 0$ in Eq.(\ref{theta-body-2}). 
\end{proof}

The proof now follows analogously to a similar statement in \cite{Shepherd, Fujie-Tamura, Marcel-2}. To show this, we use the following alternative characterization of TH(G) in \cite{Fujie-Tamura},
\begin{equation}
\label{theta-body-3}
TH(G) = \{ |\mathcal{P} \rangle \in \mathbb{R}^{|V|} : \langle \mathcal{P} | M | \mathcal{P} \rangle - \sum_{i \in |V|} M_{i,i} | \mathcal{P} \rangle_i \leq 0 \; \; \forall M \in \mathbb{M} \}, 
\end{equation}
with 
\begin{equation}
\label{M-Fujie}
\mathbb{M} \defeq \{M \in \mathbb{S}^{|V|} : M_{\textsl{u}, \textsl{v}} = 0 \; (\textsl{u} \neq \textsl{v}, \textsl{u} \nsim \textsl{v}, \; M \succeq 0 \}.
\end{equation}
Consider a non-local vertex $\mathcal{P}^{nl}$ of the no-signaling polytope $\mathbf{NS}(n,m,k)$. From Fact 1, we know that there is an associated matrix $\tilde{A}$ such that $\tilde{A} \cdot |\mathcal{P}^{nl} \rangle = | \tilde{\mathit{b}} \rangle$, i.e., such that 
\begin{eqnarray}
\label{positivity-eq}
| \mathcal{P}^{nl} \rangle_i = 0 \qquad i \in \text{Pos}, 
\end{eqnarray}
where $\text{Pos}$ is the set of non-negativity constraints which are set to equality in the vertex box, and the number of \textit{independent} non-negativity constraints is $|\text{Pos}|_{Ind} = D = \left( m(k-1) + 1 \right)^n - 1$  (in general $\prod_{i=1}^{n} \left[ m_i (k_i - 1) +1 \right] - 1$) is the dimension of the no-signaling polytope $\mathbf{NS}(n,m,k)$. In other words, the vertex of $\mathbf{NS}(n,m,k)$ has $D$ independent probabilities that are assigned value $0$. Moreover, the box also obeys normalization and no-signaling constraints which are written as maximum clique equalities
\begin{equation}
\sum_{i \in c} | \mathcal{P}^{nl} \rangle_i = 1, \qquad \; \forall c \in C_{n,ns} 
\end{equation}
and the number of such \textit{independent} constraints is $|C_{n,ns}|_{Ind} = \bar{D} - D$. 
Consider the matrix $\tilde{M} \defeq \tilde{A}^T \tilde{A}$, which can be written as
\begin{equation}
\tilde{M} = \sum_{i \in \text{Pos}} I_{(i,i)} + \sum_{c \in\text{C}_{n,ns}} |j^{c} \rangle \langle j^{c}|,
\end{equation}
where $I_{(i,i)}$ is an indicator matrix of size $\bar{D} \times \bar{D}$ with only one non-zero diagonal entry at position $(i,i)$ which equals $1$. $|j^{c} \rangle$ is a $0-1$ vector of length $\bar{D}$ with support on the clique $c \in C_{n,ns}$ (where it takes value $1$) corresponding to a normalization or no-signaling constraint. It is readily checked that $\tilde{M} \in \mathbb{M}$ from Eq.(\ref{M-Fujie}). For any non-local vertex $\mathcal{P}^{nl}$, since we have $D$ independent non-negativity constraints and $\bar{D} - D$ independent normalization and no-signaling constraints, the matrix $\tilde{A}$ is of rank $\bar{D}$ and therefore $\tilde{M}$ is positive definite. By explicit calculation, one sees that for this choice of $\tilde{M}$ ,
\begin{equation}
\langle {\mathcal{P}^{nl}} | \tilde{M} | \mathcal{P}^{nl} \rangle - \sum_{i =1}^{\bar{D}} \tilde{M}_{i,i} | \mathcal{P}^{nl} \rangle_i = 0.
\end{equation} 
Since $\mathcal{P}^{nl}$ is a non-local vertex of a rational convex polytope, it has some entry $|\mathcal{P}^{nl} \rangle_{j}$ which is non-integral (and rational), i.e., $\exists j \in [\bar{D}] \; \; \text{s.t.} \; \; | \mathcal{P}^{nl} \rangle_{j} \notin \{0,1\}$. Now, one can construct the matrix $\tilde{M}' \defeq \tilde{M} - \epsilon I_{j,j}$ which is guaranteed to be positive semi-definite for some finite $\epsilon > 0$ since $\tilde{M}$ is positive definite. Moreover, we also have $\tilde{M}' \in \mathbb{M}$ from Eq.(\ref{M-Fujie}). Now, since $0 < |\mathcal{P}^{nl} \rangle_{j} < 1$, we have 
$| \mathcal{P}^{nl} \rangle_{j} \left( | \mathcal{P}^{nl} \rangle_j - 1\right) < 0$.
Therefore, 
\begin{equation}
\langle {\mathcal{P}^{nl}} | \tilde{M}' | \mathcal{P}^{nl} \rangle - \sum_{i =1}^{\bar{D}} \tilde{M}'_{i,i} | \mathcal{P}^{nl} \rangle_i = - \epsilon |\mathcal{P}^{nl} \rangle_{j} \left( |\mathcal{P}^{nl} \rangle_j - 1 \right) > 0,
\end{equation}
so that $\mathcal{P}^{nl} \notin TH(G)$ due to the characterization of TH(G) in Eq.(\ref{theta-body-3}). As we have seen in Theorem \ref{theta-alq}, for any Bell scenario $\mathit{B}(n,m,k)$, $\mathbf{Q}(n,m,k) \subseteq \mathbf{\tilde{Q}}(n,m,k) = \text{TH}^{(c)}(G_{\mathit{B}}) \subset \text{TH}(G_{\mathit{B}})$. Therefore, $\mathcal{P}^{nl} \notin \mathbf{Q}(n,m,k)$. 

Finally, consider the closure of the quantum set denoted as $\text{cl}(\mathbf{Q}(n,m,k))$. We know that $\text{cl}(\mathbf{Q}(n,m,k)) \subseteq \mathbf{\tilde{Q}}(n,m,k)$ since the NPA hierarchy \cite{NPA} converges to the closure of the quantum set. We also know (see for example \cite{GPT}) that the set $\text{TH}(G)$ is closed. This implies from the above argument that $\mathcal{P}^{nl} \notin \text{cl}(\mathbf{Q}(n,m,k))$ so that the extremal no-signaling boxes also cannot be approximated arbitrarily closely in quantum theory.  
\end{proof}


\begin{thebibliography}{0}
\expandafter\ifx\csname natexlab\endcsname\relax\def\natexlab#1{#1}\fi
\expandafter\ifx\csname bibnamefont\endcsname\relax
  \def\bibnamefont#1{#1}\fi
\expandafter\ifx\csname bibfnamefont\endcsname\relax
  \def\bibfnamefont#1{#1}\fi
\expandafter\ifx\csname citenamefont\endcsname\relax
  \def\citenamefont#1{#1}\fi
\expandafter\ifx\csname url\endcsname\relax
  \def\url#1{\texttt{#1}}\fi
\expandafter\ifx\csname urlprefix\endcsname\relax\def\urlprefix{URL }\fi
\providecommand{\bibinfo}[2]{#2}
\providecommand{\eprint}[2][]{\url{#2}}

\bibitem{Principles1}
G. Brassard, H. Buhrman, N. Linden, A. A. Methot, A. Tapp and F. Unger,
Phys. Rev. Lett. \textbf{96}, 250401 (2006).

\bibitem{Principles2}
M. Pawlowski, T. Paterek, D. Kaszlikowski, V. Scarani, A. Winter and M. Zukowski,
Nature \textbf{461}, 1101 (2009). 

\bibitem{Principles3}
M. Navascu\'{e}s and H. Wunderlich, 
Proc. Royal Soc. A \textbf{466}: 881 (2009). 

\bibitem{Principles4}
T. Fritz, A. B. Sainz, R. Augusiak, J. B. Brask, R. Chaves, A. Leverrier and A. Ac\'in,
Nature Communications \textbf{4}, 2263 (2013). 

\bibitem{Principles5}
A. Cabello,
Phys. Rev. Lett. \textbf{110}, 060402 (2013); 

\bibitem{BHK}
J. Barrett, L. Hardy and A. Kent,
Phys. Rev. Lett. \textbf{95}, 010503 (2005).

\bibitem{CR}
R. Colbeck and R. Renner,
Nature Physics \textbf{8}, 450 (2012). 

\bibitem{Pir10}
S. Pironio et al.,
Nature 464, 1021 (2010).

\bibitem{Tsirelson}
B. S. Cirel'son, 
Lett. Math. Phys. \textbf{4}, 93 (1980).

%
%
%
%
%


\bibitem{CHSH}
 J. F. Clauser, M. A. Horne, A. Shimony and R. A. Holt, 
 Phys. Rev. Lett. \textbf{23}, 880 (1969).

\bibitem{PR}
S. Popescu and D. Rohrlich, Found. Phys. \textbf{24}, 379 (1994).

\bibitem{Fritz2}
T. Fritz,
J. Math. Phys. \textbf{53}, 072202 (2012).

\bibitem{YCATS}
T. H. Yang, D. Cavalcanti, M. Almeida, C. Teo and V. Scarani, 
New J. Phys. \textbf{14}, 013061 (2012).

\bibitem{Winter}
A. Cabello, S. Severini and A. Winter,
arXiv: 1010.2163 (2010).

\bibitem{CSW}
A. Cabello, S. Severini and A. Winter,
Phys. Rev. Lett. \textbf{112}, 040401 (2014).


\bibitem{Fritz}
A. Ac\'{i}n, T. Fritz, A. Leverrier and A. B. Sainz,
Comm. Math. Phys. \textbf{334}(2), 533 (2015).


\bibitem{Schrijver}
A. Schrijver, "Combinatorial Optimization. Polyhedra and Efficiency". Vol. A, volume 24 of Algorithms and Combinatorics. Springer-Verlag, Berlin (2003).  




\bibitem{Khot}
S. Khot, 
Proceedings of the thirty-fourth annual ACM symposium on Theory of computing, 767 (2002).

\bibitem{Cleve}
R. Cleve, P. Hoyer, B. Toner and J. Watrous, 
arXiv: 0404076 (2004). 

\bibitem{BS2009}
N. Brunner and P. Skrzypczyk, Phys. Rev. Lett. \textbf{102}, 160403 (2009).

\bibitem{Pironio}
S. Pironio,
J. Math. Phys. 46, 062112 (2005).

\bibitem{EW2014}
H. Ebbe and S. Wolf,  IEEE Trans. on Inf. Theory, \textbf{60}(2), 1159 (2014).

\bibitem{BLMPPR2005}
J. Barrett, N. Linden, S. Massar, S. Pironio, S. Popescu and D. Roberts, Phys. Rev. A \textbf{71}, 022101 (2005).

\bibitem{GHZ}
D. M. Greenberger, M. A. Horne, and A. Zeilinger, in 
\textit{Bell's Theorem, Quantum Theory, and Conceptions of the Universe}
(Kluwer, Dordrecht), 69 (1989).

\bibitem{BC}
S. L. Braunstein and C. M. Caves, Annals of Physics {\bf 202}, 22 (1990).

\bibitem{NPA}
M. Navascu\'{e}s, S. Pironio and A. Ac\'{i}n,
New J. Phys. \textbf{10}, 073013 (2008).

\bibitem{AQ}
M. Navascu\'{e}s, Y. Guryanova, M. J. Hoban and A. Ac\'{i}n,
Nature Communications \textbf{6}, 6288 (2015).


\bibitem{Lovasz-2}
M. Gr\"{o}tschel, L. Lov\'{a}sz and A. Schrijver, J. Combin. Theory Ser. B, \textbf{40}(3), 330 (1986).  


\bibitem{Fujie-Tamura}
T. Fujie and A. Tamura, J. Op. Res. Soc. Jpn. 45, 285
(2002).


\bibitem{Shepherd}
F. B. Shepherd, In \textit{Perfect graphs}, Wiley-Intersci. Ser. Discrete Math. Optim., 261, Wiley, Chichester (2001).



\bibitem{Marcel-2}
M. K. de Carli Silva, PhD Thesis,
University of Waterloo (2013). 

\bibitem{Horn}
R. A. Horn and C. R. Johnson,
\textit{Matrix Analysis}, Cambridge University Press, Cambridge (1987). 

\bibitem{Julio}
J. I. de Vicente,
Phys. Rev. A \textbf{92}, 032103 (2015).

\bibitem{RRG}
R. Ramanathan, R. Augusiak and G. Murta,
Phys. Rev. A \textbf{93}, 022333 (2016). 

\bibitem{Barrett}
J. Barrett, arXiv: quant-ph/0508211 (2005).

\bibitem{BG14}
S. Beigi and A. Gohari, arXiv: 1409.3665 (2014).


%
%
%
%


\bibitem{GPT}
J. Gouveia, P. A. Parrilo and R. R. Thomas,
SIAM J. Optim. Vol. 20, 4, 2097 (2010). \\

%






%
%
%
%
%




































\end{thebibliography}
\end{document}